\documentclass[a4paper,fleqn,usenatbib]{mnras}

\usepackage[T1]{fontenc}
\usepackage{ae,aecompl}

\usepackage{graphicx}	
\usepackage{amsmath}	
\usepackage{amssymb}	
\usepackage{bbm}
\usepackage{wasysym}
\usepackage{txfonts}

\DeclareMathAlphabet{\mathpzc}{OT1}{pzc}{bx}{it}

\usepackage{mathrsfs}
\DeclareMathAlphabet{\mathscrbf}{OMS}{mdugm}{b}{n}

\newcommand{\oder}[2]{\dfrac{\mathrm{d}{#1}}{\mathrm{d}{#2}}}

\newcommand{\myvec}[1]{\mathbfit{#1}}
\newcommand{\mymatrix}[1]{\mathbfss{#1}}

\newcommand{\ignore}[1]{}

\newcommand{\myquat}[1]{\mathsf{#1}}
\newcommand{\cliff}{\mathfrak{Cl}}

\newcounter{thetheorem}
\newenvironment{theorem}{
\refstepcounter{thetheorem}
\vspace{2mm}\noindent\textbf{Theorem \arabic{thetheorem}:}\itshape}{\\[-1mm] }

\newcounter{thedefinition}

\newcommand{\mansol}{\Gamma}

\newcommand{\celmech}{Celest. Mech.}
\newcommand{\cmda}{Celest. Mech. Dyn. Astron.}

\title[Stability in KS space and the Hopf fibration]{Stability and chaos in Kustaanheimo-Stiefel space induced by the Hopf fibration}

\author[J. Roa et al.]{
Javier Roa,$^1$\thanks{\emph{Present address:} Jet Propulsion Laboratory, California Institute of Technology, 4800 Oak Grove Drive, Pasadena, CA 91109-8099, USA. E-mail: javier.roa@upm.es (JR)}
Hodei Urrutxua$^{2,1}$ and
Jes\'us Pel\'aez$^1$
\\
$^1$Space Dynamics Group, Technical University of Madrid, Pza. Cardenal Cisneros 3, Madrid, E-28040, Spain\\
$^2$Astronautics Research Group, University of Southampton, Highfield, SO17 1BJ, UK
}

\date{Accepted 2016 March 31. Received 2016 January 25; in original form 2015 September 14\\doi: \href{http://dx.doi.org/10.1093/mnras/stw780}{10.1093/mnras/stw780}\vspace{-3mm}}

\pubyear{2016}

\begin{document}
\label{firstpage}
\pagerange{\pageref{firstpage}--\pageref{lastpage}}
\maketitle

\begin{abstract}
The need for the extra dimension in Kustaanheimo-Stiefel (KS) regularization is explained by the topology of the Hopf fibration, which defines the geometry and structure of KS space. A trajectory in Cartesian space is represented by a four-dimensional manifold, called the fundamental manifold. Based on geometric and topological aspects classical concepts of stability are translated to KS language. The separation between manifolds of solutions generalizes the concept of Lyapunov stability. The dimension-raising nature of the fibration transforms fixed points, limit cycles, attractive sets, and Poincar\'e sections to higher-dimensional subspaces. From these concepts chaotic systems are studied. In strongly perturbed problems the numerical error can break the topological structure of KS space: points in a fiber are no longer transformed to the same point in Cartesian space. An observer in three dimensions will see orbits departing from the same initial conditions but diverging in time. This apparent randomness of the integration can only be understood in four dimensions. The concept of topological stability results in a simple method for estimating the time scale in which numerical simulations can be trusted. Ideally all trajectories departing from the same fiber should be KS transformed to a unique trajectory in three-dimensional space, because the fundamental manifold that they constitute is unique. By monitoring how trajectories departing from one fiber separate from the fundamental manifold a critical time, equivalent to the Lyapunov time, is estimated. These concepts are tested on $N$-body examples: the Pythagorean problem, and an example of field stars interacting with a binary.
\end{abstract}




\section{Introduction}

In the 1960's Eduard Stiefel started to organize the Oberwolfach Meetings on Celestial Mechanics, in an attempt to draw the interest of mathematicians into this subject. The first of those meetings took place in 1964 and Paul Kustaanheimo presented his work on describing Keplerian motion using spinors. His work on spinors combined with Stiefel's experience in regularization gave birth to the celebrated extension of the Levi-Civita transformation to the three-dimensional case, known as the Kustaanheimo-Stiefel transformation \citep{kustaanheimo1965perturbation}. This extension had eluded researchers since \citet{levi1920regularisation} presented his original regularization of the planar problem and \citet{hurwitzueber} proved that transformations of this type only exist in spaces of dimension $n=1,2,4,8$. This statement was further developed by \citet{adams1966k}.

Heinz \citet{hopf1931abbildungen} discovered a particular transformation from the unit 3-sphere, $\mathcal{S}^3$, onto the unit 2-sphere, $\mathcal{S}^2$, so that the preimage of each point in three-dimensional space turns out to be a circle on $\mathcal{S}^3$, called a fiber. All points in this fiber transform into the same point in three-dimensional space. Such transformation is referred to as the Hopf fibration. In fact, the Kustaanheimo-Stiefel transformation can be understood as a particular Hopf map \citep[][\S44]{stiefel1971linear}. Hopf was Stiefel's doctoral advisor and influenced other areas of his research, including the Hopf-Stiefel functions and the Stiefel manifolds. \citet{davtyan1987generalised} developed the generalization of KS transformation to the case $\mathbb{R}^8\to\mathbb{R}^5$ in order to transform the problem of the five-dimensional hydrogen atom into an eight-dimensional oscillator. They successfully rewrote the Hamiltonian of the hydrogen atom as the Hamiltonian of an eight-dimensional isotrope oscillator. \citet{deprit1994linearization} published an exhaustive treaty on the transformations underlying KS regularization. They focused on the topic of linearization, connecting with prior work from Lagrange.

The KS transformation provides a robust regularization scheme for dealing with close approaches or even impact trajectories. Close encounters between stars are one of the major challenges in $N$-body simulations. The first extensions of the KS transformation to $N$-body problems are due to \citet{peters1968treatment}, \citet{aarseth1971direct} and \citet{bettis1971treatment}. We refer to the work of \citet{szebehely1971recent} for a review of the methods developed in those years. \citet{aarseth1974regularization} focused on the three-body problem and their method was later generalized by \citet{heggie1974global}, who reformulated the Hamiltonian for dealing with an arbitrary number of particles. The shortcoming of Heggie's method is that it fails to reproduce collisions of more than two particles. \citet{mikkola1985practical} discovered a technique for avoiding this singularity by rewriting the Sundman time transformation in terms of the Lagrangian of the system. This method integrates $4N(N-1)+1$ equations of motion, so its use is recommended for few-body problems. The formulations based on the KS transformation have been improved throughout the years \citep{mikkola1998efficient,mikkola2008implementing}, especially since the development of the chain regularization techniques \citep{mikkola1989chain,mikkola1993implementation}. The introduction of relativistic corrections in the models has occupied different authors: \citet{kupi2006dynamics} modified the KS regularization for two-body close encounters in $N$-body simulations by introducing post-Newtonian effects; \citet{funato1996time} published a reformulation of the KS transformation focused on time-symmetric algorithms. \citet{aarseth1999nbody1,aarseth2003gravitational} presented several reviews of the evolution and keystones in the development of $N$-body simulations. The Levi-Civita variables have recently been recovered by \citet{lega2011detection} to detect resonant close encounters in the three-body problem, and used by \citet{astakhov2004capture} in combination with an extension of the phase space to analyze the elliptic restricted three-body problem.

Binary systems of stars appear naturally in simulations of star clusters and galaxies. The dynamical interaction between binaries and field stars is a challenging problem both from the theoretical and computational perspective due to its chaotic nature. The pioneering numerical exploration of stellar dynamics by \citet{aarseth1963dynamical,aarseth1966dynamical} was followed by \citet{heggie1975binary}, who analyzed different configurations of binary encounters using KS regularization. In a series of papers \citet{hut1983binary} and \citet{hut1993binary} focused on the gravitational scattering of field stars interacting with stellar binaries. They relied on extensive Monte Carlo simulations, and papers in the same series focused on deriving analytic solutions to the problem \citep{hut1983binaryII,heggie1993binary}. \citet{kiseleva1994note} addressed the problem of the stability of triple stars, seeking empirical formulas for modeling the time when the system becomes unstable. Stellar collisions between binary systems have been studied by \citet{fregeau2004stellar}. \citet{tout1997rapid} published an efficient algorithm for the simulation of additional physical phenomena occurring in stellar binaries. In state of the art $N$-body codes such as \textsc{nbody6} binaries and close two-body encounters are analyzed using KS regularization based on the Stumpff functions \citep{mikkola1998efficient}, as described by the author \citep{aarseth1999nbody1}.

The topology of the KS transformation has motivated many studies on the subject. \citet{velte1978concerning} explored its representation in the language of quaternions, a task that also occupied \citet{vivarelli1983ks} and \citet{waldvogel2006order,waldvogel2006quaternions}. Different representations of KS regularization have been published by \citet{vivarelli1986kepler,vivarelli1994ks}, including a representation in hypercomplex algebra \citep{vivarelli1985ks}. Hypercomplex numbers have recently been recovered by \citet{roa2015orbit} to derive a regularization scheme that describes orbital dynamics in the geometry of Minkowski space-time. \citet{deprit1994linearization} re-derived the KS transformation by doubling, using quaternions and octonions. They conclude by deriving the Burdet-Ferr\'andiz transformation from the foundations of KS regularization \citep{ferrandiz1987general}. \citet{saha2009interpreting} has recently reformulated the problem by combining quaternions with the Hamiltonian formalism. \citet{elbialy2007kustaanheimo} approached the KS transformation from an alternative perspective, focusing on the connection with the Hopf map and the topological structure of the transformation. In a series of papers \citet{fukushima2003efficient,fukushima2004efficient,fukushima2005efficientTE} analyzed different numerical aspects of the regularization, seeking scaling factors that guarantee the conservation of the integrals of motion and time-elements for improving the stability of the time transformation.


In this paper we seek a theory of stability in KS space based on the physics of the problem. We generalize the concepts of fixed points, limit cycles, attractors, Lyapunov and structural stability, and Poincar\'e maps to four-dimensions by means of the KS transformation. By synchronizing the relative dynamics of nearby trajectories in terms of the physical (and not fictitious) time, conclusions about the stability of the system derived in KS variables apply also to Cartesian. From these definitions we advance to chaotic systems and show how numerical errors can destroy the topological structure of KS space: ideally every point in a fiber transforms to the same point in Cartesian space (so it is typically chosen arbitrarily); but for strongly perturbed problems trajectories emanating from the same fiber may separate in time and no longer represent the same trajectory in three-dimensions. $N$-body problems are the main application we consider, but the results can be extended to any application of KS regularization.

In Section~\ref{Sec:KS_Hopf} the connection between the KS transformation and the Hopf fibration is established. The required equations are presented, emphasizing the geometric aspects. Section~\ref{Sec:stability} is devoted to defining the concept of fundamental manifold, and to showing how Lyapunov and Poincar\'e stability are understood in four dimensions. In Section~\ref{Sec:chaos} we study the exponential divergence of trajectories in KS space and derive a simple method to estimate an indicator that is equivalent to the Lyapunov time. Finally, practical examples are presented in Section~\ref{Sec:scattering}.

\section{The KS transformation as a Hopf map}\label{Sec:KS_Hopf}
Let $\mathbfit{r}=(x,y,z)^\top$ be the position vector of a point in Cartesian space $\mathbb{E}^3$, projected in an inertial frame $\mathfrak{I}$, and let $\myvec{x}=(x,y,z,0)^\top$ be its extension to $\mathbb{R}^4$. \citet{kustaanheimo1965perturbation} found a regularization of the two-body problem introducing the new coordinates $\myvec{u}=(u_1,u_2,u_3,u_4)^\top$, defined in the parametric space $\mathbb{U}^4$ embedded in $\mathbb{R}^4$. The KS transformation is defined explicitly as 
\begin{equation}\label{Eq:KS}
	\mathbfit{x}= \mathscr{K}(\myvec{u}) =\mathbfss{L}(\mathbfit{u})\,\mathbfit{u}
\end{equation}
where $\mathbfss{L}(\mathbfit{u})$ is known as the KS matrix:
\begin{equation}
	\mathbfss{L}(\mathbfit{u}) = \left[
		\begin{array}{rrrr}
			u_1 & -u_2 & -u_3 &  u_4  \\ 
			u_2 &  u_1 & -u_4 & -u_3  \\ 
			u_3 &  u_4 &  u_1 &  u_2  \\ 
			u_4 & -u_3 &  u_2 & -u_1  
		\end{array}
	\right]
\end{equation}
The KS matrix is $r$-orthogonal, i.e.
\begin{equation}\label{Eq:orthogonal}
	\mathbfss{L}^{-1}(\mathbfit{u}) = \frac{1}{r}\,\mathbfss{L}^\top(\mathbfit{u})
\end{equation}
Every point $\mathbfit{u}$ is KS-mapped to one single point in Cartesian space $\mathbb{E}^3$. These equations are a particular case of the more general map proposed by \citet{hopf1931abbildungen}. 

Regularizing the equations of orbital motion by means of the KS transformation requires the time transformation due to \cite{sundman1912memoire}:
\begin{equation}\label{Eq:sundman}
	\mathrm{d}t = r\,\mathrm{d}s
\end{equation}
where $s$ is referred to as the fictitious time, and $r=||\mathbfit{r}||$. Derivatives with respect to physical time $t$ will be denoted by a dot, $\dot{\mathbfit{r}}$, whereas derivatives with respect to fictitious time will be written $\mathbfit{r}^\prime$. The radial distance $r$ relates to the KS variables by means of
\begin{equation}\label{Eq:radius}
	r = u_1^2+u_2^2+u_3^2+u_4^2 = ||\mathbfit{u}||^2
\end{equation}
The KS transformation maps fibers on the 3-sphere of radius $\sqrt{r}$ in $\mathbb{U}^4$ to points on the 2-sphere of radius $r$ in $\mathbb{E}^3$.

\citet{hopf1931abbildungen} proved that the transformation from the 3-sphere to the 2-sphere maps circles to single points, defining the structure $\mathcal{S}^1\hookrightarrow\mathcal{S}^3\to\mathcal{S}^2$. Equation~\eqref{Eq:KS} is invariant under the gauge transformation $\mathscr{R}:\mathbfit{u}\mapsto\mathbfit{w}$, 
\begin{equation}\label{Eq:identity}
	\mathbfit{x} = \mathbfss{L}(\mathbfit{u})\,\mathbfit{u} = 	\mathbfss{L}(\mathbfit{w})\,\mathbfit{w}
\end{equation}
Vector $\mathbfit{w}=(w_1,w_2,w_3,w_4)^\top$ takes the form:
\begin{equation}\label{Eq:fiber_matrix}
	\mathbfit{w} = \mathscr{R}(\vartheta;\myvec{u}) = \mathbfss{R}(\vartheta)\,\mathbfit{u}
\end{equation}
where $\mathbfss{R}(\vartheta)$ is the matrix 
\begin{equation}\label{Eq:matrix_R}
	\mathbfss{R}(\vartheta) = \left[ \begin{array}{cccc}
		\cos\vartheta & 0 & 0 & -\sin\vartheta \\
		0 & \cos\vartheta & \sin\vartheta & 0 \\
		0 & -\sin\vartheta & \cos\vartheta & 0 \\
		\sin \vartheta & 0 & 0 & \cos\vartheta
\end{array}	 \right]
\end{equation}
This matrix is orthogonal, and also
\begin{equation}
	\mathbfss{R}^\top(\vartheta) = \mathbfss{R}(-\vartheta) 
\end{equation}
Being $\mathbfss{R}(\vartheta)$ orthogonal Eq.~\eqref{Eq:fiber_matrix} can be inverted to provide
\begin{equation}\label{Eq:inverse_fiber}
	\mathbfit{u} = \mathscr{R}^{-1}(\vartheta;\myvec{w}) = \mathbfss{R}(-\vartheta)\,\mathbfit{w}
\end{equation}
The transformation $\mathscr{R}$ preserves the radius $r$, i.e.
\begin{equation}
	r = \myvec{u}\cdot\myvec{u} = \myvec{w}\cdot\myvec{w}
\end{equation}
Since the radius is invariant to the selection of the point in the fiber it follows that the physical time, defined by Eq.~\eqref{Eq:sundman}, is $\mathscr{R}$-invariant as well\footnote{Alternative forms of the time transformation can be found in the literature, generalized as $\mathrm{d}t/\mathrm{d}s=g(\myvec{x},\dot{\myvec{x}})$. We refer to the work by \citet{zare1975time} for a survey of transformations involving different powers of the radial distance, the potential, the Lagrangian, or combinations of the relative separations for the case of $N$-body problems. The vectors $\myvec{x}$ and $\dot{\myvec{x}}$ are $\mathscr{R}$-invariant, so the uniqueness of the physical time is also guaranteed for more general transformations.}. The identity in Eq.~\eqref{Eq:identity} and the $r$-orthogonality of matrix $\mymatrix{L}$ furnish a useful relation:
\begin{equation}
	\myvec{w} = \mymatrix{L}^{-1}(\myvec{w})\,\mymatrix{L}(\myvec{u})\,\myvec{u} = \mymatrix{R}(\vartheta)\,\myvec{u}\implies \mymatrix{L}^\top(\myvec{w})\,\mymatrix{L}(\myvec{u}) = r\mymatrix{R}(\vartheta)
\end{equation}

The angular variable $\vartheta$ parameterizes the Hopf fibration in four-dimensional space. In fact, Eq.~\eqref{Eq:fiber_matrix} defines explicitly the fiber $\mathcal{F}$: changing the value of $\vartheta$ defines different points in $\mathbb{U}^4$ that are KS transformed to the same point in $\mathbb{E}^3$. This yields the definition of fiber as the subset of all points in four-dimensional space that are mapped into the same point in $\mathbb{E}^3$ by means of the KS transformation,
\begin{equation}
	\mathcal{F} = \big\{ \mathbfit{w}(\vartheta)\in\mathbb{U}^4\,\big|\,\mathbfit{x}=\mathscr{K}(\mathbfit{w}), \,\forall \vartheta\in[0,2\pi] \big\}
\end{equation}
A different fiber transforms into a different point. Consequently, two fibers cannot intersect because the intersection point will then be transformed into the same point in $\mathbb{E}^3$ despite belonging to two different fibers \citep[][p.~271]{stiefel1971linear}. The stereographic projection of the fibers onto $\mathbb{E}^3$ \citep[see for example][]{griguolo2012correlators} reveals that two fibers in KS space are connected by a Hopf link, as sketched in Fig. \ref{Fig:link}.
\begin{figure}
	\centering
	\includegraphics[width=.55\linewidth]{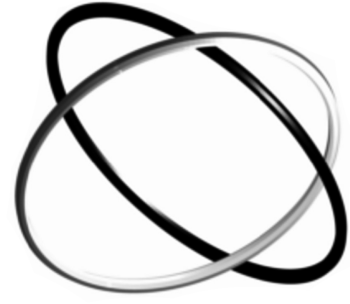}
	\caption{Hopf link connecting two different fibers in KS space. The Hopf fibration is visualized by means of the stereographic projection to $\mathbb{E}^3$.\label{Fig:link}}
\end{figure}

\subsection{The velocity and the bilinear relation}\label{Sec:velocity}
Let $\mathbfit{u},\mathbfit{w}\in\mathbb{U}^4$. The KS matrix satisfies the property
\begin{equation}
	\mathbfss{L}(\mathbfit{u})\,\mathbfit{w} = \mathbfss{L}(\mathbfit{w})\,\mathbfit{u} \iff \ell(\mathbfit{u},\mathbfit{w})=0
\end{equation}
where $\ell(\mathbfit{u},\mathbfit{w})$ denotes the bilinear relation
\begin{equation}
	\ell(\mathbfit{u},\mathbfit{w}) = u_1w_4 - u_2w_3 + u_3w_2 - u_4w_1
\end{equation}

Differentiating Eq.~\eqref{Eq:KS} with respect to fictitious time and taking into account the time transformation from Eq.~\eqref{Eq:sundman} yields
\begin{equation}\label{Eq:velocity}
	\dot{\mathbfit{x}} = \frac{2}{r}\,\mathbfss{L}(\mathbfit{u})\,\mathbfit{u}^\prime
\end{equation}
where $\dot{\myvec{x}}=(v_x,v_y,v_z,0)^\top$ is the velocity vector extended to $\mathbb{R}^4$. Note that the fourth component is zero, which means
\begin{equation}
	\ell(\mathbfit{u},\mathbfit{u}^\prime) = 0
\end{equation}
Moreover, \citet[][p.~29]{stiefel1971linear} proved that $\ell(\mathbfit{u},\mathbfit{u}^\prime)=0$ is a first integral of orbital motion. Provided that the KS transformation is $\mathscr{R}$-invariant, it follows that the bilinear relation holds for all points in a given fiber,
\begin{equation}\label{Eq:bilin_V_W}
	\ell(\mathbfit{w},\mathbfit{w}^\prime) = \ell(\mathbfit{u},\mathbfit{u}^\prime) = 0
\end{equation}

Let $\mathbfit{t}(\vartheta)\in\mathbb{U}^4$ denote the vector that is tangent to a fiber $\mathcal{F}$ at $\mathbfit{w}(\vartheta)$. The direction of $\mathbfit{t}$ can be obtained by differentiating Eq.~\eqref{Eq:fiber_matrix} with respect to $\vartheta$. It reads
\begin{equation}
	\mathbfit{t} = \mathbfss{R}^\ast(\vartheta)\,\mathbfit{u}
\end{equation}
where $ \mathbfss{R}^\ast(\vartheta)=\mathrm{d}\mathbfss{R}(\vartheta)/\mathrm{d}\vartheta$ is obtained by differentiating Eq.~\eqref{Eq:matrix_R}.
Taking as an example $\vartheta=0$ yields the components of the tangent vector $\mathbfit{t}$,
\begin{equation}\label{Eq:tangent}
	\mathbfit{t} = (-u_4,u_3,-u_2,u_1)
\end{equation}
This unveils a geometric interpretation of the bilinear relation $\ell(\mathbfit{u},\mathbfit{v})=0$: it can be understood as an orthogonality condition \citep[][\S43]{stiefel1971linear}, since
\begin{equation}
	\ell(\mathbfit{u},\mathbfit{v})=0\iff\mathbfit{v}\cdot\mathbfit{t}=0
\end{equation} 
Two vectors $\mathbfit{u}$ and $\mathbfit{v}$ satisfy the bilinear relation $\ell(\mathbfit{u},\mathbfit{v})=0$ if $\mathbfit{v}$ is orthogonal to the fiber through $\mathbfit{u}$. Provided that $\ell(\mathbfit{u},\mathbfit{u}^\prime)=0$ holds naturally and it is an integral of motion it follows that the velocity in KS space, $\mathbfit{u}^\prime$, is always orthogonal to the fiber at $\mathbfit{u}$. The fiber bundle $\mathcal{S}^1\hookrightarrow\mathcal{S}^3\to\mathcal{S}^2$ shows that the fibers constituting the 3-sphere are circles, corresponding to points on the 2-sphere. Indeed, the tangent vector $\mathbfit{t}(\vartheta)$ is always perpendicular to the position vector $\mathbfit{w}(\vartheta)$,
\begin{equation}
	\mathbfit{w}\cdot\mathbfit{t} = \big[ \mathbfss{R}(\vartheta)\,\mathbfit{u} \big]\cdot\left[  \mathbfss{R}^\ast(\vartheta)\,\mathbfit{u} \right] = \myvec{u}\cdot\left\{ \mymatrix{R}^\top(\vartheta) \left[  \mathbfss{R}^\ast(\vartheta)\myvec{u} \right] \right\} =0
\end{equation}
no matter the value of $\vartheta$. Appendix~\ref{Sec:Appendix} is devoted to the definition of orthogonal bases and a cross product in $\mathbb{U}^4$.

\subsection{The inverse mapping}
The inverse KS transformation $\mathscr{K}^{-1}:\mathbfit{x}\mapsto\mathbfit{u}$ maps points to fibers. Introducing the auxiliary vector $\myvec{v}=(v_1,v_2,v_3,v_4)^\top$ the inverse mapping takes the form
\begin{equation}\label{Eq:inverse_KS}
	\begin{aligned}
		 v_1 & = R \sin\theta \\
		 v_2 & = \frac{1}{2R}(y\sin\theta - z\cos\theta) \\
		 v_3 & = \frac{1}{2R}(y\cos\theta + z\sin\theta) \\
		 v_4 & = -R\cos\theta
	\end{aligned}
\end{equation}
Here $R^2=(r+|x|)/2$. The angle $\theta$ is different from $\vartheta$: the points on the fiber are parameterized by $\theta$, which is measured with respect to a certain axis; given two points $\mathbfit{u}$ and $\mathbfit{w}$ obtained by setting $\theta=\theta_1$ and $\theta_2$ in Eq.~\eqref{Eq:inverse_KS}, respectively, they relate by virtue of Eq.~\eqref{Eq:fiber_matrix}. This equation then provides the relation:
\begin{equation}
	\theta_2 - \theta_1 = \vartheta 
\end{equation}
meaning that the variable $\vartheta$ denotes the angular \emph{separation} between points along the same fiber. The value of $\theta$ depends on the position of the reference axis, whereas $\vartheta$ is independent from the selection of the axis.

The point $\mathbfit{u}$ is finally defined as
\begin{equation}\label{Eq:ic1}
\begin{aligned}
	\mathbfit{u}=(v_1,v_2,v_3,v_4)^\top \quad & \text{if}\;x\geq0 \\
	\mathbfit{u}=(v_2,v_1,v_4,v_3)^\top \quad & \text{if}\;x<0
\end{aligned}
\end{equation}
Two alternative expressions are considered for avoiding potential singularities. They differ in the selection of the axes in KS space. From this result any point $\mathbfit{w}_0$ in the initial fiber $\mathcal{F}_0$ can be obtained from 
\begin{equation}\label{Eq:initial_w0}
\begin{aligned}
	& \mathbfit{w}_0(\vartheta) =  \mathbfss{R}(\vartheta)\,\mathbfit{u}_0\quad &&\text{if}\;x_0\geq0\\
	& \mathbfit{w}_0(\vartheta) =  \mathbfss{R}(-\vartheta)\,\mathbfit{u}_0\quad &&\text{if}\;x_0<0
\end{aligned}
\end{equation}
so that $\mathbfit{x}_0=\mathbfss{L}(\mathbfit{w}_0)\mathbfit{w}_0$. The sign criterion complies with the different definitions of the axes in KS space. See Appendix~\ref{Sec:Appendix} for a discussion on the orthogonal frames attached to the fiber.

The velocity in $\mathbb{U}^4$ is obtained by inverting Eq.~\eqref{Eq:velocity}, taking into account the orthogonality relation in Eq.~\eqref{Eq:orthogonal}:
\begin{equation}
	\mathbfit{u}^\prime = \frac{1}{2}\mathbfss{L}^\top(\mathbfit{u})\,\dot{\mathbfit{x}}
\end{equation}

The geometry of the inverse KS transformation can be studied from Fig. \ref{Fig:fibers}. The grey sphere is three-dimensional and of radius $r$. The black arc corresponds to a set of initial conditions, $\mathbfit{r}_j$. The white dot represents one particular position in $\mathbb{E}^3$, $\mathbfit{r}_i$. The inverse KS transformation applied to $\mathbfit{r}_i$ yields the fiber $\mathcal{F}_i$. The fiber is represented by means of its stereographic projection to $\mathbb{R}^3$. The black surface consists of all the fibers $\mathcal{F}_j$ that are KS mapped to the points $\mathbfit{r}_j$. In this figure it is possible to observe the Hopf link connecting different fibers.
\begin{figure}
	\centering
	\includegraphics[width=\linewidth]{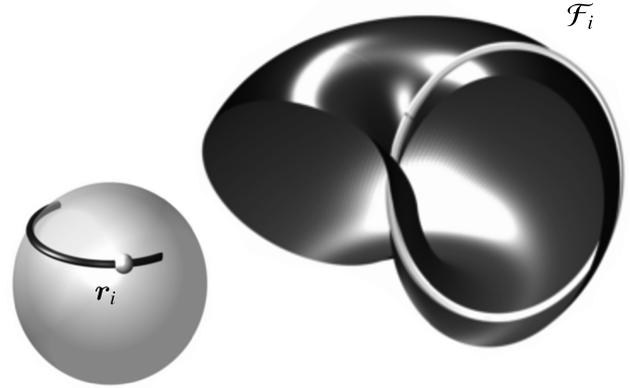}
	\caption{Stereographic projection to $\mathbb{R}^3$ of the Hopf fibration corresponding to a set of initial positions on the three-dimensional sphere of radius $r$. The black semi-torus consists of all the fibers that transform into the semi-circumference on the 2-sphere on the bottom-left corner. One single fiber $\mathcal{F}_i$ is plotted in white, corresponding to $\mathbfit{r}_i$.\label{Fig:fibers}}
\end{figure}


\section{Stability in KS space}\label{Sec:stability}

The classical concepts of stability from Lyapunov and Poincar\'e can be translated to KS language by considering the topology of the transformation. First, we introduce an important theorem regarding the geometry of the fibers. From this theorem the concept of the fundamental manifold arises naturally. 

The stability concepts here presented are not based in numerical analyses; previous studies about the stability of KS transformation \citep{baumgarte1972numerical,baumgarte1976stabilized,arakida2000long} focus on the behavior of the \emph{numerical procedure}. We aim for a series of definitions that capture the \emph{physical behavior}, that should be independent from the formulation of the dynamics.

\subsection{A central theorem}

Two fibers can never intersect, as discussed when formally defining a fiber. It is now possible to advance on this statement and to formulate a fundamental property of the KS transformation: 

\begin{theorem}\label{Thm:angle} (Conservation of $\vartheta$)
The angular separation between two trajectories emanating from $\mathcal{F}_0$, measured along every fiber, is constant. That is
\begin{equation}
	\mathbfit{w}_0 = {\text{\normalfont \bfseries \textsf{R}}}(\vartheta_0)\,\mathbfit{u}_0\implies\mathbfit{w}(s) = {\text{\normalfont \bfseries \textsf{R}}}(\vartheta_0)\,\mathbfit{u}(s)
\end{equation}
for any value of $\vartheta_0$ and the fictitious time $s$. This is an intrinsic property of KS space and does not depend on the dynamics of the system.
\end{theorem}

\begin{proof}
Consider two trajectories in KS space, $\mathbfit{u}=\mathbfit{u}(s)$ and $\mathbfit{w}=\mathbfit{w}(s)$, departing from the same fiber $\mathcal{F}_0$. They relate by means of Eq.~\eqref{Eq:fiber_matrix}. In the most general case the angle $\vartheta$ can be described by a function $\vartheta=\vartheta(s)$ and initially it is $\vartheta(0)=\vartheta_0$. The trajectories evolve according to
\begin{equation}\label{Eq:ws_proof}
	\mathbfit{w}(s) = \mathscr{R}(\vartheta;\myvec{u}(s)) = \mathbfss{R}(\vartheta)\mathbfit{u}(s)
\end{equation}
Differentiating this equation with respect to fictitious time yields
\begin{equation}\label{Eq:dwds_proof}
	\mathbfit{w}^\prime(s) = \mathbfss{R}^\prime(\vartheta)\,\mathbfit{u}(s) + \mathbfss{R}(\vartheta)\,\mathbfit{u}^\prime(s)
\end{equation}
Equation~\eqref{Eq:bilin_V_W} proved that the bilinear relation holds for any trajectory in KS space, meaning that $\ell(\mathbfit{w},\mathbfit{w}^\prime)=\ell(\mathbfit{u},\mathbfit{u}^\prime)=0$. This renders:
\begin{equation}
	\ell(\mathbfit{w},\mathbfit{w}^\prime) = \ell\big(\mathbfss{R}(\vartheta)\mathbfit{u},\mathbfss{R}^\prime(\vartheta)\,\mathbfit{u} + \mathbfss{R}(\vartheta)\,\mathbfit{u}^\prime\big) = 0
\end{equation}
after substituting Eqs.~\eqref{Eq:ws_proof} and \eqref{Eq:dwds_proof}. Expanding the bilinear relation in the previous expression shows that
\begin{alignat}{2}
	\ell\big(\mathbfss{R}(\vartheta)\mathbfit{u},\mathbfss{R}^\prime(\vartheta)\,\mathbfit{u} + \mathbfss{R}(\vartheta)\,\mathbfit{u}^\prime\big) = r\,\oder{\vartheta}{s} + \ell(\mathbfit{u},\mathbfit{u}^\prime)  = 0
\end{alignat}
Assuming that $r>0$ and considering that $\ell(\mathbfit{u},\mathbfit{u}^\prime)=0$ one gets
\begin{equation}
	\oder{\vartheta}{s} = 0\implies \vartheta(s) = \vartheta_0 
\end{equation}
so the angular separation along every fiber remains constant. We emphasize that no assumptions about the dynamics have been made.
\end{proof}

A direct consequence of this result is the relation between the velocities along the trajectories $\myvec{u}(s)$ and $\myvec{w}(s)$:
\begin{equation}
	\myvec{w}^\prime(s) = \mathscr{R}(\vartheta;\myvec{u}^\prime(s))
\end{equation}

\subsection{The fundamental manifold $\mansol$}
A trajectory in Cartesian space, understood as a continuum of points in $\mathbb{E}^3$, is represented by a continuum of fibers in $\mathbb{U}^4$. Each fiber is KS transformed to a point of the trajectory. The fibers form the \emph{fundamental manifold}, $\mansol$.

Equation~\eqref{Eq:initial_w0} defines the initial fiber $\mathcal{F}_0$, which yields a whole family of solutions parameterized by the angular variable $\vartheta$. Every trajectory $\myvec{w}(s)$ is confined to the fundamental manifold. Thanks to Thm.~\ref{Thm:angle} the manifold $\mansol$ can be constructed following a simple procedure: first, a reference trajectory $\myvec{u}(s)$ is propagated from any point in $\mathcal{F}_0$; then, mapping the transformation $\mathscr{R}$ over it renders a fiber $\mathcal{F}_i$ for each point $\myvec{u}(s_i)$ of the trajectory. The set $\cup_i\mathcal{F}_i$ defines $\mansol$. Recall that 
\begin{equation}
	\bigcap_i\mathcal{F}_i = \emptyset
\end{equation}
The fact that all trajectories emanating from $\mathcal{F}_0$ are confined to $\mansol$ is what makes an arbitrary choice of $\theta$ in Eq.~\eqref{Eq:inverse_KS} possible. The diagram in Fig.~\ref{Fig:diagram} depicts the construction of the fundamental manifold $\mansol$. 
\begin{figure}
	\centering
	\includegraphics[width=.65\linewidth]{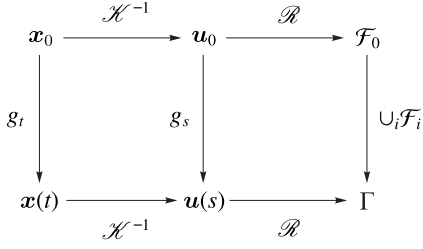}
	\caption{Construction of the fundamental manifold. The mapping $g_t:\myvec{x}_0\mapsto\myvec{x}(t)$ denotes the integration of the trajectory from $t_0$ to $t$. Similarly, $g_s$ refers to the propagation using the fictitious time.\label{Fig:diagram}}
\end{figure}

\subsection{Fixed points, limit cycles and attractors}
Points in $\mathbb{E}^3$ transform into fibers in $\mathbb{U}^4$. Thus, a fixed point in Cartesian space, $\myvec{x}_0$, translates into a fixed fiber in KS space, $\mathcal{F}_0$. Asymptotically stable fixed fibers (to be defined formally in the next section) attract the fundamental manifold of solutions, $\mansol\to\mathcal{F}_0$. Asymptotic instability is equivalent to the previous case under a time reversal.

Limit cycles are transformed to fundamental manifolds, referred to as limit fundamental manifolds $\mansol_0$. A fundamental manifold $\mansol$ originating in the basin of attraction of a limit fundamental manifold will converge to it after sufficient time. For $\mansol\to\mansol_0$ convergence means that each fiber in $\mansol$ approaches the corresponding fiber in $\mansol_0$. Correspondence between fibers is governed by the $t$-synchronism.

In a more general sense, attractors in $\mathbb{U}^4$ are invariant sets of the flow. The point-to-fiber correspondence connects attractors in $\mathbb{E}^3$ with attractors in KS space. The basin of attraction of an attractive set $Y_u\subset\mathbb{U}^4$ is built from its definition in 3-dimensions. Let $X\subset\mathbb{E}^3$ be the basin of attraction of $Y$. It can be transformed to KS space, $X\to X_u$, thanks to
\begin{equation}
	X_u = (\mathscr{R}\circ\mathscr{K}^{-1})(X) = \mathscr{R}\big(\mathscr{K}^{-1}(X) \big) 
\end{equation}
This construction transforms arbitrary sets in $\mathbb{E}^3$ to $\mathbb{U}^4$. The inverse KS transformation constitutes a dimension raising mapping, so in general $\dim(X_u)=\dim(X)+1$.


\subsection{Relative dynamics and synchronism}
The theories about the local stability of dynamical systems are based on the relative dynamics between nearby trajectories. The concepts of stability formalize how the separation between two (initially close) trajectories evolves in time. But the concept of \emph{time evolution} requires a further discussion because of having introduced an alternative time variable via the Sundman transformation.

Keplerian motion is known to be Lyapunov unstable \citep[see][for example]{baumgarte1972numerical}. Small differences in the semimajor axes of two orbits result in a separation that grows in time because of having different periods. However, Kepler's problem transforms into a harmonic oscillator by means of the KS transformation, with the fictitious time being equivalent to the eccentric anomaly. The resulting system is stable: for fixed values of the eccentric anomaly the separation between points in each orbit will be small, because of the structural (or Poincar\'e) stability of the motion. These considerations are critical for the numerical integration of the equations of motion. But in this paper we seek a theory of stability in $\mathbb{U}^4$ expressed in the language of the physical time $t$, because of its physical and practical interest. The conclusions about the stability of the system will be equivalent to those obtained in Cartesian space.

The spectrum of the linearized form of Kepler's problem written in Cartesian coordinates,
\begin{equation}
	\oder{^2\myvec{r}}{\,t^2} = -\frac{\myvec{r}}{r^3}
\end{equation}
exhibits one eigenvalue with positive real part, $\lambda=\sqrt{2/r^3}$. Lyapunov's theory of linear stability states that the system is unstable.

Under the action of the KS transformation Kepler's problem transforms into
\begin{equation}\label{Eq:kepler_KS}
	\oder{^2\myvec{u}}{s^2} = -\frac{h}{2}\,\myvec{u}
\end{equation}
where $h$ is minus the Keplerian energy. Although the linear analysis is not useful in this case, selecting a candidate Lyapunov function $V(\myvec{u},\myvec{u}^\prime)=h(\myvec{u}\cdot\myvec{u})/4+(\myvec{u}^\prime\cdot\myvec{u}^\prime)/2$ the stability of the system is proved. In order to represent the Lyapunov instability of the motion with respect to time $t$ the Sundman transformation needs to be considered. Given two circular orbits of radii $r_1$ and $r_2$, the time delay between both solutions reads
\begin{equation}
	\Delta t = t_2 - t_1 = (r_2-r_1)s
\end{equation}
The time delay grows with fictitious time and small values of $r_2-r_1$ do not guarantee that $\Delta t$ remains small. 

This phenomenon relates to the synchronism of the solutions \citep{roa2015error,roa2016theoryI}. Solutions to the system defined in Eq.~\eqref{Eq:kepler_KS} are stable if they are synchronized in fictitious time, but unstable if they are synchronized in physical time. We adopt this last form of synchronism for physical coherence.

\subsection{Stability of the fundamental manifold}

\subsubsection{Lyapunov stability}
A trajectory $\myvec{r}(t)$ in $\mathbb{E}^3$ is said to be Lyapunov stable if, for every small $\varepsilon>0$, there is a value $\delta>0$ such that for any other solution $\myvec{r}^\ast(t)$ satisfying $||\myvec{r}(t_0)-\myvec{r}^\ast(t_0)||<\delta$ it is $||\myvec{r}(t)-\myvec{r}^\ast(t)||<\varepsilon$, with $t>t_0$. In KS language trajectory translates into fundamental manifold. In order to extend the definition of Lyapunov stability accordingly an adequate metric $d$ to measure the distance between manifolds is required.

Let $\mansol_1$ and $\mansol_2$ be two (distinct) fundamental manifolds. The fibers in $\mansol_1$ can never intersect the fibers in $\mansol_2$. But both manifolds may share certain fibers, corresponding to the points of intersection between the two resulting trajectories in Cartesian space. The distance between the manifolds at $t\equiv t(s_1)=t(s_2)$ is the distance between the corresponding fibers. Setting $\theta$ to a reference value $\theta_\mathrm{ref}$ in Eq.~\eqref{Eq:inverse_KS} so that $\theta_1=\theta_2\equiv\theta_{\mathrm{ref}}$, we introduce the metric:
\begin{equation}\label{Eq:metric}
	%
	d(t;\mansol_1,\mansol_2) = \frac{1}{2\pi}\int_0^{2\pi}||\myvec{w}_1(s_1;\vartheta)-\myvec{w}_2(s_2;\vartheta)||\,\mathrm{d}\vartheta
\end{equation}
with $d(t;\mansol_1,\mansol_2)\equiv d(\mathcal{F}_1,\mathcal{F}_2)$. It is measured by computing the distance between points in $\mansol_1$ and $\mansol_2$ with the same value of $\vartheta$, and then integrating over the entire fiber. It is defined for given values of physical time, and not fictitious time. The reason is that the goal of this section is to define a theory of stability such that the fundamental manifold inherits the stability properties of the trajectory in Cartesian space. This theory is based on the physics of the system, not affected by a reformulation of the equations of motion.

Consider a fundamental manifold $\mansol$, referred to a nominal trajectory $\myvec{r}(t)$, and a second manifold $\mansol^\ast$ corresponding to a perturbed trajectory $\myvec{r}^\ast(t)$. If the nominal trajectory is Lyapunov stable, then for every $\varepsilon_u>0$ there is a number $\delta_u>0$ such that
\begin{equation}
	d(t_0;\mansol,\mansol^\ast) < \delta_u \implies d(t;\mansol,\mansol^\ast)<\varepsilon_u
\end{equation}
If the initial separation between the manifolds is small it will remain small according to the metric defined in Eq.~\eqref{Eq:metric}. 

The nominal solution $\myvec{r}(t)$ is said to be asymptotically stable if $||\myvec{r}(t)-\myvec{r}^\ast(t)||\to0$ for $t\to\infty$. Similarly, the fundamental manifold $\mansol$ will be asymptotically stable if $d(t;\mansol,\mansol^\ast)\to0$ for sufficiently large times. The opposite behavior $d(t;\mansol,\mansol^\ast)\to\infty$ corresponds to an asymptotically unstable fundamental manifold. It behaves as if it were asymptotically stable if the time is reversed.

 
\subsubsection{Poincar\'e maps and orbital stability}
The notion of Poincar\'e (or orbital) stability is particularly relevant when analyzing the fundamental manifold due to its geometric implications. Kepler's problem is unstable in the sense of Lyapunov but it is orbitally stable: disregarding the time evolution of the particles within their respective orbits, the separation between the orbits remains constant. 

The definition of the Poincar\'e map in $\mathbb{E}^3$ involves a 2-dimensional section $\Sigma$ that is transversal to the flow. Denoting by $\myvec{p}_1$, $\myvec{p}_2$,\ldots the successive intersections of a periodic orbit with $\Sigma$, the Poincar\'e map $\mathscr{P}$ renders
\begin{equation}
	\mathscr{P}(\myvec{p}_n) = \myvec{p}_{n+1}
\end{equation}
The generalization of the Poincar\'e section to KS space $\mathscr{K}:\Sigma\to\Sigma_u$ results in a subspace embedded in $\mathbb{U}^4$. In Sec.~\ref{Sec:velocity} we showed that the trajectories intersect the fibers at right angles, provided that the velocity $\myvec{u}^\prime$ is orthogonal to the vector tangent to the fiber. Thus, every fiber defines a section that is transversal to the flow. The transversality condition for $\Sigma$ translates into the section containing the fiber at $\myvec{u}$.

The Poincar\'e section $\Sigma_u$ can be constructed by combining the set of fibers that are KS transformed to points in $\Sigma$. Let $\myvec{n}=(n_x,n_y,n_z)^\top$ be the unit vector normal to $\Sigma$ in $\mathbb{E}^3$, projected onto an inertial frame. The Poincar\'e section takes the form
\begin{equation}\label{Eq:Sigma_poincare}
	\Sigma\equiv n_x (x-x_0) + n_y (y-y_0) + n_z (z-z_0) = 0
\end{equation}
where $(x_0,y_0,z_0)$ are the coordinates of the first intersection point. Equation~\eqref{Eq:Sigma_poincare} can be written in parametric form as $\Sigma(x(\eta,\xi),y(\eta,\xi),z(\eta,\xi))$, with $\eta$ and $\xi$ two free parameters. The extended Poincar\'e section $\Sigma_u$ is obtained by transforming points on $\Sigma$ to KS space and then mapping the fibration $\mathscr{R}$:
\begin{equation}
	\Sigma_u =  (\mathscr{R}\circ\mathscr{K}^{-1})(\Sigma) 
\end{equation}
The choice of the Poincar\'e section $\Sigma$ is not unique, and therefore the construction of $\Sigma_u$ is not unique either. The resulting Poincar\'e section $\Sigma_u$ is a subspace of dimension three embedded in $\mathbb{U}^4$. Indeed, the transformation $(\mathscr{R}\circ\mathscr{K}^{-1})(\Sigma)$ provides:
\begin{equation}
	\Sigma\mapsto \Sigma_u(u_1(\eta,\xi,\vartheta),u_2(\eta,\xi,\vartheta),u_3(\eta,\xi,\vartheta),u_4(\eta,\xi,\vartheta))
\end{equation}
meaning that points in $\Sigma_u$ are fixed by three parameters, $(\eta,\xi,\vartheta)$. The dimension is raised by $(\mathscr{R}\circ\mathscr{K}^{-1})$.

The intersection between a given fundamental manifold and the Poincar\'e section $\Sigma_u$ results in a fiber,
\begin{equation}
	\mansol\cap\Sigma_u=\mathcal{F}
\end{equation}
Successive intersections can be denoted $\mathcal{F}_1$, $\mathcal{F}_2$,\ldots. The Poincar\'e map in $\mathbb{U}^4$, $\mathscr{P}:\Sigma_u\to\Sigma_u$, is
\begin{equation}
	\mathscr{P}(\mathcal{F}_n) = \mathcal{F}_{n+1}
\end{equation}  
Every point in a fiber intersects $\Sigma_u$ simultaneously. Due to the $\mathscr{R}$-invariance of the Sundman transformation the time period between crossings is the same for every trajectory connecting $\mathcal{F}_n$ and $\mathcal{F}_{n+1}$.

Let $\mansol$ denote a fundamental manifold representing a nominal periodic orbit, and let $\mansol^\ast$ be a perturbed solution. They differ in the conditions at the first $\Sigma$-crossing, $\mathcal{F}_1$ and $\mathcal{F}_1^\ast$ respectively. The manifold $\mansol$ is said to be Poincar\'e (or orbitally) stable if
\begin{equation}
	d(\mathcal{F}_1^\ast,\mathcal{F}_1)<\delta_u \implies d(\mathscr{P}^n(\mathcal{F}_1^\ast),\mathcal{F}_1)<\varepsilon_u
\end{equation}
If the separation between the fibers at the first crossing is small, the separation will remain small after $n$ crossings.


\section{Order and chaos}\label{Sec:chaos}
In the previous section we generalized the key concepts of dynamical stability to KS space. The approach we followed aims for a theory that captures the physical properties of the system, instead of focusing on its purely numerical conditioning. The next step is the analysis of chaos in $\mathbb{U}^4$.

Chaotic systems are extremely sensitive to numerical errors due to the strong divergence of the integral flow. This is specially important in the vicinity of singularities, and it is precisely here where KS regularization exhibits all its potential. This section focuses on characterizing the exponential divergence of trajectories in $\mathbb{U}^4$ due to highly unstable dynamics. 


By definition the fundamental manifold is mapped to a trajectory in $\mathbb{E}^3$. The equations of motion in $\mathbb{U}^4$ are no more than a reformulation of a dynamical system originally written in $\mathbb{E}^3$. For sufficiently smooth perturbations the Picard-Lindel\"of theorem ensures the uniqueness of the solution. Thus, the corresponding fundamental manifold is also unique and its KS transform defines only one trajectory. This means that any trajectory in the fundamental manifold is mapped to the same exact trajectory in $\mathbb{E}^3$, no matter the position within the initial fiber. An observer in three-dimensional space, unaware of the extra degree of freedom introduced by the gauge $\mathscr{R}$, will always perceive the same trajectory no matter the values of $\vartheta$.

\subsection{The $\mathscr{K}$-separation}
In order to integrate the equations of motion numerically in $\mathbb{U}^4$ the initial values of $\myvec{u}_0$ and $\myvec{u}_0^\prime$ need to be fixed. This means choosing a point in the fiber $\mathcal{F}_0$. Since all the points in $\mathcal{F}_0$ are KS transformed to the same exact state vector in $\mathbb{E}^3$, the selection of the point is typically arbitrary. But for an observer in $\mathbb{U}^4$ different values of $\vartheta$ yield different initial conditions, and therefore the initial value problem to be integrated may behave differently. Ideally\footnote{Due to the limited precision of floating point arithmetic, even the fact that all points generated with Eq.~\eqref{Eq:initial_w0} and varying $\vartheta$ will be KS-transformed to the same exact point in $\mathbb{E}^3$ should be questioned. The loss of accuracy in the computation of the initial conditions in $\mathbb{U}^4$ will eventually introduce errors of random nature. As a result, Eq.~\eqref{Eq:initial_w0} provides points that are not exactly in the true fiber. Although the separation is small (of the order of the round-off error) and negligible in most applications, it may have an impact on the numerical integration of chaotic systems.} all trajectories emanating from $\mathcal{F}_0$ remain in the same fundamental manifold, that is unique. However, numerical errors leading to the exponential divergence of the trajectories can cause the trajectories to depart from the fundamental manifold. In other words, after sufficient time two trajectories originating from the same fiber $\mathcal{F}_0$, $\myvec{w}_0=\mathscr{R}(\vartheta;\myvec{u}_0)$, will no longer define the same fiber $\mathcal{F}(s)$, $\myvec{w}(s)\neq\mathscr{R}(\vartheta;\myvec{u}(s))$. In this case Thm.~\ref{Thm:angle} will be violated. Multiple fundamental manifolds will appear, obtained by mapping the transformation $\mathscr{R}$ over each of the trajectories. The observer in $\mathbb{E}^3$ will see a collection of trajectories that depart from the same exact state vector and they separate in time, as if the problem had a random component. This behavior can only be understood in four dimensions.

These topological phenomena yield a natural way of measuring the error growth in KS space without the need of a precise solution. Let $\myvec{u}(s)$ be a reference trajectory in $\mathbb{U}^4$, and let $\myvec{w}(s)$ be a second trajectory defined by $\myvec{w}_0=\mathscr{R}(\vartheta;\myvec{u}_0)$. It is possible to build the fundamental manifold $\mansol$ from the solution $\myvec{u}(s)$. The second solution is expected to be $\myvec{w}^\ast(s)=\mathscr{R}(\vartheta;\myvec{u}(s))$ by virtue of Thm.~\ref{Thm:angle}. When numerical errors are present $\myvec{w}(s)$ and its expected value $\myvec{w}^\ast(s)$ (the projection of the fundamental manifold) may not coincide. Note that $\myvec{w}(s)=\myvec{w}^\ast(s)$ ensures the uniqueness of the solution, but says nothing about its accuracy. The separation between $\myvec{w}(s)$ and its projection on $\mansol$ is an indicator of the breakdown of the topological structure supporting the KS transformation, meaning that the solutions can no longer be trusted.

Motivated by this discussion we introduce the concept of the $\mathscr{K}$-separation, $d_\mathscr{K}$
\begin{equation}\label{Eq:K-separation}
	d_\mathscr{K}(s) = || \mathbfit{w}(s) - \mathbfit{w}^\ast(s)  || = || \mathbfit{w}(s) - \mathscr{R}(\vartheta;\mathbfit{u}^\ast(s))  ||
\end{equation} 
defined as the Euclidean distance between an integrated trajectory and its projection on the manifold of solutions. Monitoring the growth of the $\mathscr{K}$-separation is a way of quantifying the error growth of the integration. In the context of $N$-body simulations, \cite{quinlan1992reliability} discuss how the separation between nearby trajectories evolves: the divergence is exponential in the linear regime when the separation is small, but the growth rate is reduced when the separation is large. At this point the separation might be comparable to the interparticle distance. The $\mathscr{K}$-separation will grow exponentially at first (for $d_\mathscr{K}\ll1$) until it is no longer small ($d_\mathscr{K}\sim\mathcal{O}(1)$), and then its growth slows down. Locating the transition point is equivalent to finding the time scale $t_\mathrm{cr}$ in which the solution in KS space can no longer be trusted: for $t<t_\mathrm{cr}$ the topological structure of $\mathbb{U}^4$ is preserved, but for $t>t_\mathrm{cr}$ the uniqueness of the manifold of solutions $\mansol$ is not guaranteed.

For $t<t_\mathrm{cr}$ the $\mathscr{R}$-invariance of the Sundman transformation holds. The time for all the points in a fiber coincides. Thus, $t_\mathrm{cr}$ and $s_\mathrm{cr}$ are interchangeable: at $t<t_\mathrm{cr}$ it is also $s<s_\mathrm{cr}$. The behavior of the solutions can be equally analyzed in terms of the physical or the fictitious time.

In practice the $\mathscr{K}$-separation is evaluated as follows:
\begin{enumerate}
	\item Choosing a reference $\theta$ in Eqs.~\eqref{Eq:inverse_KS} and \eqref{Eq:ic1}, for example $\theta=0$, integrate $\myvec{u}^\ast(s)$.
	\item Propagate a second trajectory $\myvec{w}(s)$ generated from Eq.~\eqref{Eq:initial_w0} with $\vartheta\neq0$.
	\item Build the expected trajectory $\myvec{w}^\ast(s)$ by mapping $\mathscr{R}(\vartheta)$ over $\myvec{u}^\ast(s)$, i.e. $\myvec{w}^\ast=\mymatrix{R}(\vartheta)\myvec{u}^\ast$. The $\mathscr{K}$-separation is the Euclidean distance between $\myvec{w}(s)$ and $\myvec{w}^\ast(s)$.
\end{enumerate}

\subsection{Topological stability}
The uniqueness of $\mansol$ can be understood as \emph{topological stability}. KS space is said to be topologically stable if all the trajectories emanating from the same fiber define a unique manifold of solutions, and therefore they are all KS-transformed to the same trajectory in $\mathbb{E}^3$. For an observer in $\mathbb{E}^3$ a topologically unstable system seems non-deterministic, with solutions departing from the same initial conditions but separating in time with no apparent reason. 

A system is topologically stable in the interval $t<t_\mathrm{cr}$. The trajectories diverge exponentially,
\begin{equation}
	d_\mathscr{K}(t)/d_\mathscr{K}(0)\sim\mathrm{e}^{\gamma_t t}\qquad \text{or}\qquad d_\mathscr{K}(s)/d_\mathscr{K}(0)\sim\mathrm{e}^{\gamma_s s}
\end{equation}
Here $\gamma$ is equivalent to a Lyapunov exponent. For $t>t_\mathrm{cr}$ this equation no longer models the growth of the $\mathscr{K}$-separation and the system is topologically unstable. Simulations over the transition time $t_\mathrm{cr}$ integrated in $\mathbb{U}^4$ can no longer be trusted. Depending on the integrator, the integration tolerance, the floating point arithmetic, the compiler, etc. the values of $t_\mathrm{cr}$ for a given problem might change. Thus, topological stability is a property of a certain propagation, which requires all the previous factors to be defined.

The validity of the solution for an integration over the critical time $t_\mathrm{cr}$ is not guaranteed. When $t_\mathrm{cr}<t_\mathrm{esc}$ (with $t_\mathrm{esc}$ denoting the escape time) not even the value of $t_\mathrm{esc}$ can be estimated accurately. In such a case solutions initialized at different points in the fiber may yield different escape times.  

The method presented in this section provides an estimate of the interval in which the propagation is topologically stable. The exponent $\gamma$ depends on the integration scheme and the dynamics, but it is not strongly affected by the integration tolerance. An estimate of the value of $\gamma$ provides an estimate of the critical time for a given integration tolerance $\varepsilon$. Assuming $d_\mathscr{K}(t_\mathrm{cr})\sim1$:
\begin{equation}\label{Eq:estimate_tcr}
	t_\mathrm{cr}\sim -\frac{1}{\gamma_t}\log\varepsilon
\end{equation} 
Conversely, if the simulation needs to be carried out up to a given $t_f$, the required integration tolerance is approximately
\begin{equation}
	\varepsilon \sim \mathrm{e}^{-\gamma_t t_f}
\end{equation}
This simple criterion proves useful for tuning and evaluating the numerical integration. In the following examples of application the values of $\gamma_t$ are estimated by finding the slope of the exponential growth of the $\mathscr{K}$-separation in logarithmic scale. Although more rigorous algorithms could be developed, this approximation provides a good estimate of transition time between regimes.




\section{Topological stability in \textit{N}-body problems}\label{Sec:scattering}
Two examples of $N$-body problems are analyzed in this section. The first example is the Pythagorean three-body problem. The second example is a non-planar configuration of the four-body problem. This problem simulates the dynamics of two field stars interacting with a stellar binary. The experiments are designed for showing the practical aspects of the new concept of stability introduced in this paper: the topological stability of KS space.

The problems are integrated using the regularization of the $N$-body problem based on the KS transformation proposed by \citet{mikkola1985practical} as a reformulation of the method by \citet{heggie1974global}. The initialization of the method is modified so that different points on the initial fiber $\mathcal{F}_0$ can be chosen. This means generalizing the relative coordinates $\mathbfit{u}_{ij}$ to $\mathbfit{w}_{ij}$ by means of the transformation $\mathscr{R}:\mathbfit{u}\mapsto\mathbfit{w}$. The trajectories depart from the same initial conditions in $\mathbb{E}^3$. We inherit the normalization proposed in the referred paper, so that the gravitational constant is equal to one. 

Heggie-Mikkola's method is implemented in Fortran. As \citeauthor{mikkola1985practical} recommends, the problem is integrated with the \citet{bulirsch1966numerical} extrapolation scheme \citep[\S II.9]{hairer1991solvingI}. The total $\mathscr{K}$-separation is computed by combining the $\mathscr{K}$-separations for the relative dynamics of each pair of bodies. Writing $\myvec{u}_{ij}\equiv\myvec{u}_\ell$ it is 
\begin{equation}
	d_\mathscr{K} = \sqrt{\sum_{\ell}d^2_{\mathscr{K},\ell}}
\end{equation}
where $d_{\mathscr{K},\ell}$ is the $\mathscr{K}$-separation computed for $\myvec{u}_\ell$.

\subsection{The Pythagorean three-body problem}
Originally developed by \citet{burrau1913numerische}, the Pythagorean problem consists in three particles of masses $m_1=3$, $m_2=4$, and $m_3=5$. The particles will be denoted \textsf{(1)}, \textsf{(2)}, and \textsf{(3)}. At $t=0$ the bodies are at rest and lying on the vertices of a Pythagorean right triangle of sides 3, 4, and 5. The initial conditions are summarized in Table~\ref{Tab:pythagorean}. This problem has been solved and discussed in detail by \citet{szebehely1967complete}, so the solution to the problem is known. 

\begin{table}
	\caption{Dimensionless initial configuration of the Pythagorean problem. ``Id'' refers to the identification index of each body.\label{Tab:pythagorean}}
	\begin{tabular*}{\linewidth}{@{\extracolsep{\fill}}lrrrrrr@{}}
		\hline\noalign{\smallskip}
		Id &  $x$ & $y$ & $z$ & $v_x$ & $v_y$ & $v_z$ \\
		\noalign{\smallskip}\hline\noalign{\smallskip}
		\textsf{(1)} 	& 1.0000 &  3.0000 &  0.0000 &  0.0000 & 0.0000 &  0.0000 \\
		\textsf{(2)} 	&-2.0000 & -1.0000 &  0.0000 &  0.0000 & 0.0000 &  0.0000 \\
		\textsf{(3)} 	& 1.0000 & -1.0000 &  0.0000 &  0.0000 & 0.0000 &  0.0000 \\
		\noalign{\smallskip}\hline
	\end{tabular*}
\end{table}

The solution is displayed in Fig.~\ref{Fig:pythagorean}. Initially the bodies approach the origin and after a number of close-approaches body \textsf{(1)} is ejected along a trajectory in the first quadrant, whereas \textsf{(2)} and \textsf{(3)} form a binary that escapes in the opposite direction. The escape occurs at approximately $t_\mathrm{esc}\sim60$. The solution shown in the figure is obtained by setting the integration tolerance to $\varepsilon=10^{-13}$.

\begin{figure}
	\centering
	\includegraphics[width=\linewidth]{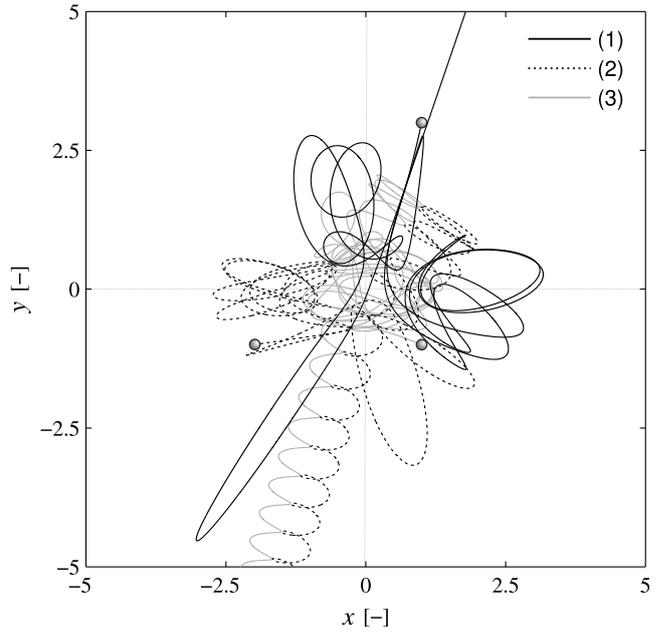}
	\caption{Solution to the Pythagorean three-body problem. The thick dots represent the initial configuration of the system.\label{Fig:pythagorean}}
\end{figure}

The problem is first integrated from a set of initial conditions obtained with $\theta=0$ in Eq.~\eqref{Eq:inverse_KS}. Then, a second trajectory initialized with $\theta\equiv\vartheta=120\degr$ is integrated and their $\mathscr{K}$-separation is shown in Fig.~\ref{Fig:Ksep_pythagorean}. After a transient the separation grows exponentially with $\gamma_t\sim5/12$, and no transitions are observed until the escape time ($t_\mathrm{cr}>t_\mathrm{esc}$). As discussed in the previous section this is equivalent to saying that the $\mathscr{K}$-separation remains small, and consequently the integration in $\mathbb{U}^4$ is topologically stable. The transformed solution in $\mathbb{E}^3$ will be unique no matter the initial position in the fiber $\mathcal{F}_0$.
\begin{figure}
	\centering
	\includegraphics[width=\linewidth]{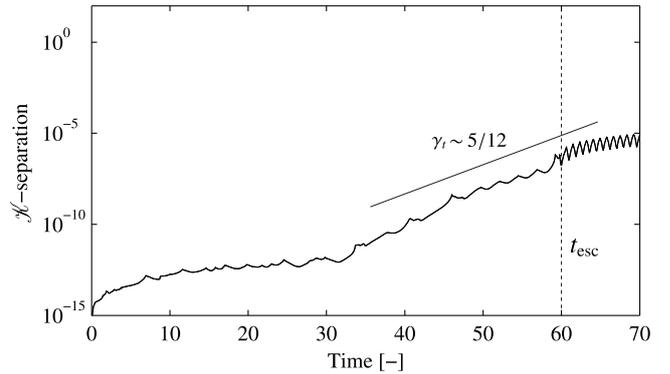}
	\caption{$\mathscr{K}$-separation for the Pythagorean problem computed from a reference trajectory with $\theta=0$ and $\vartheta=120\degr$.\label{Fig:Ksep_pythagorean}}
\end{figure}

\subsection{Field stars interacting with a stellar binary}

This second example analyzes the gravitational interaction of a binary system \textsf{(1,2)}, of masses $m_1=m_2=5$, with two incoming field stars \textsf{(3)} and \textsf{(4)} of masses $m_3=m_4=3$. The initial conditions, presented in Table \ref{Tab:scattering0}, have been selected so that both field stars reach the binary simultaneously. 

\begin{table}
	\caption{Dimensionless initial conditions for the binary system, \textsf{(1,2)}, and the field stars, \textsf{(3)} and \textsf{(4)}. ``Id'' refers to the identification index of each star.\label{Tab:scattering0}}
	\begin{tabular*}{\linewidth}{@{\extracolsep{\fill}}lrrrrrr@{}}
		\hline\noalign{\smallskip}
		Id &  $x$ & $y$ & $z$ & $v_x$ & $v_y$ & $v_z$ \\
		\noalign{\smallskip}\hline\noalign{\smallskip}
		\textsf{(1)} 	& 0.6245 &  0.6207 &  0.0000 & -0.7873 & 0.0200 & -0.0100 \\
		\textsf{(2)} 	& 0.6245 & -0.6207 &  0.0000 &  0.7873 & 0.0200 &  0.0100 \\
		\textsf{(3)} 	& 3.0000 &  3.0000 &  3.0000 & -0.3000 &-0.3000 & -0.3000 \\
		\textsf{(4)} 	&-5.0817 & -3.0000 & -3.0000 &  0.3000 & 0.2333 &  0.3000 \\
		\noalign{\smallskip}\hline
	\end{tabular*}
\end{table}

The manifold of solutions is constructed from a reference solution with $\theta=90\degr$. A second solution with $\theta=120\degr$ (or $\vartheta=30\degr$) is propagated and the corresponding $\mathscr{K}$-separation is plotted in Fig.~\ref{Fig:Ksep_4stars}. The are two different regimes in the growth of the $\mathscr{K}$-separation: the first part corresponds to the linear regime where the $\mathscr{K}$-separation is small, whereas in the second part the separation is no longer small. Both regimes are separated by {$t_\mathrm{cr}\sim42$}, when solutions in $\mathbb{E}^3$ no longer coincide. This result is in good agreement with the value predicted by Eq.~\eqref{Eq:estimate_tcr}, which is $t_\mathrm{cr}\sim40$. Since for $t=t_\mathrm{cr}$ the bodies have not yet escaped and the integration continues, the solution is topologically unstable. The escape time associated to the reference solution, {$t_\mathrm{esc}\sim75$}, might not be representative because it corresponds to the interval $t>t_\mathrm{cr}$.

\begin{figure}
	\centering
	\includegraphics[width=\linewidth]{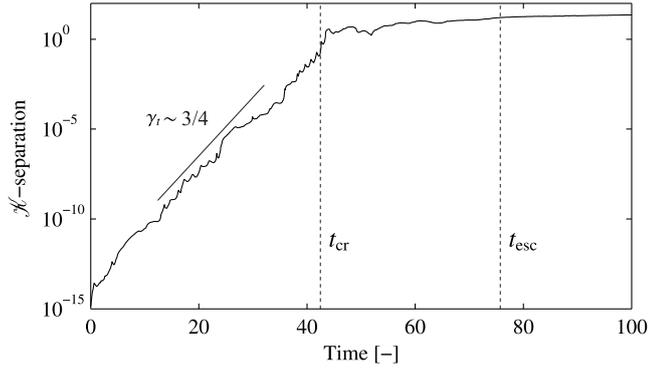}
	\caption{$\mathscr{K}$-separation for the four-body problem computed from a reference trajectory with $\theta=90\degr$ and a second trajectory with $\vartheta=30\degr$.\label{Fig:Ksep_4stars}}
\end{figure}

A direct consequence of the topological instability of the integration is the fact that solutions departing from the initial fiber $\mathcal{F}_0$ no longer represent the same solution in $\mathbb{E}^3$. Figure~\ref{Fig:solutions_4BP} shows two solutions that emanate from different points of the initial fiber. Ideally they should coincide exactly; but because the integration is topologically unstable for $t>t_\mathrm{cr}$ the difference between both solutions becomes appreciable and the accuracy of the integration over $t_\mathrm{cr}$ cannot be guaranteed. 

\begin{figure}
	\centering
	\includegraphics[width=\linewidth]{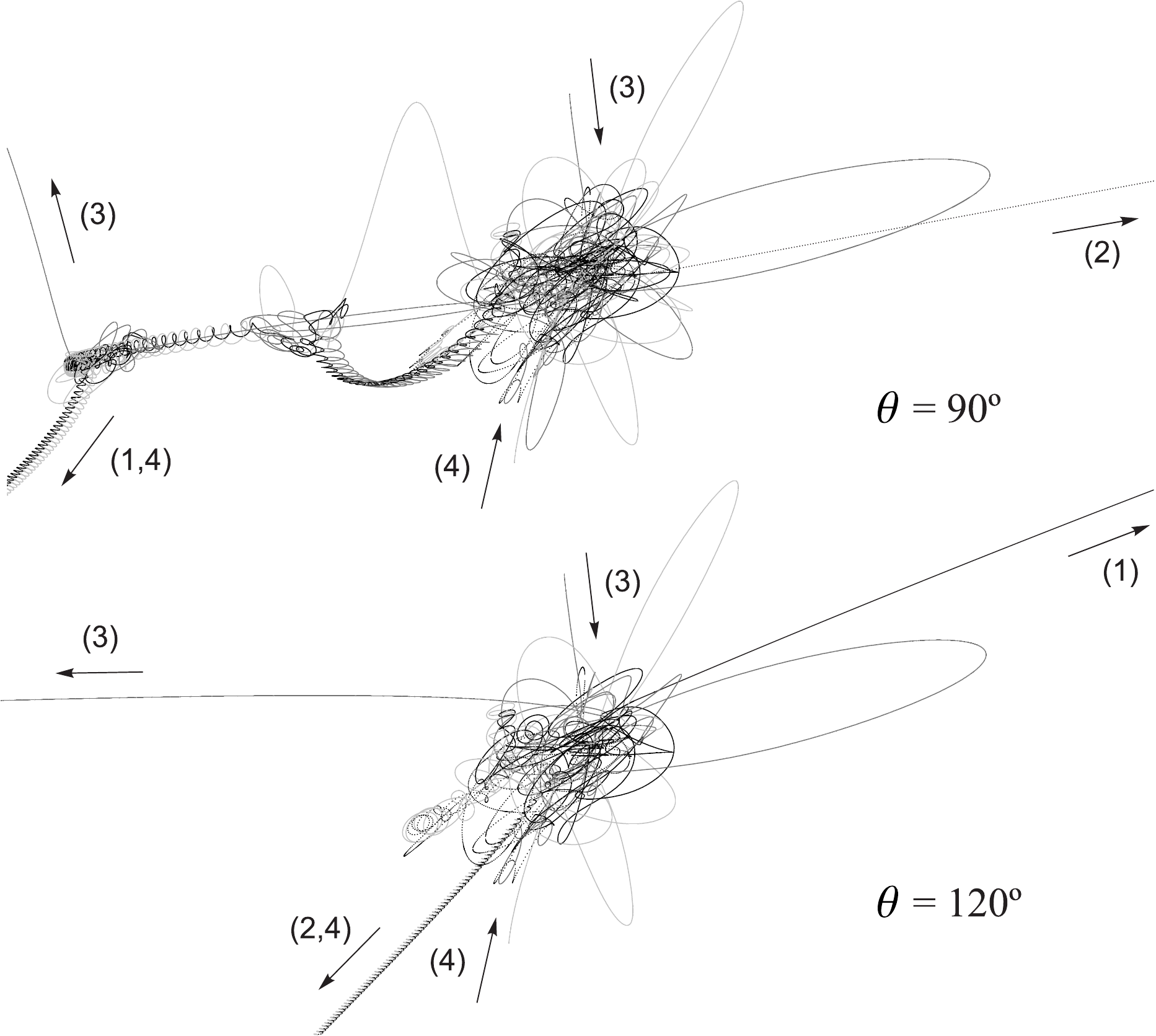}
	\caption{Two solutions to the four-body problem departing from the same fiber $\mathcal{F}_0$: the top figure corresponds to $\theta=90\degr$, and the bottom figure has been generated with $\theta=120\degr$.\label{Fig:solutions_4BP}}
\end{figure}

The topological instability is not directly related to the conservation of the energy. Although for $t>t_\mathrm{cr}$ the integration becomes topologically unstable, Fig.~\ref{Fig:energy_4BP} shows that the energy is conserved down to the integration tolerance until $t_\mathrm{esc}$, well beyond $t_\mathrm{cr}$. This is a good example of the fact that the preservation of the integrals of motion is a necessary but not sufficient condition for concluding that a certain integration is correct. 
\begin{figure}
	\centering
	\includegraphics[width=\linewidth]{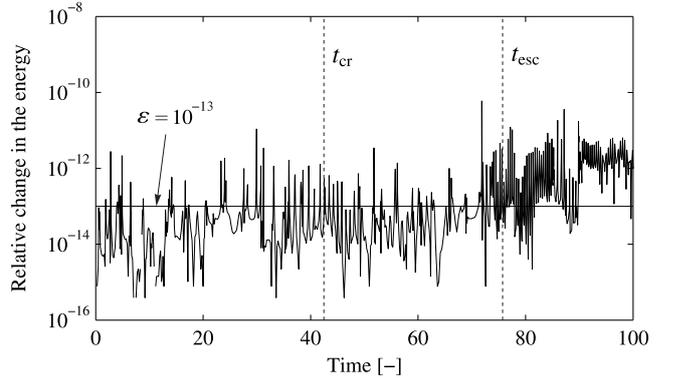}
	\caption{Relative change in the energy referred to its initial value, $(E(t)-E_0)/E_0$.\label{Fig:energy_4BP}}
\end{figure}


The evolution of the $\mathscr{K}$-separation depends on the integration scheme and the tolerance. In order to analyze this dependency Fig.~\ref{Fig:Ksep_4stars_tols} shows the results of integrating the problem with four different tolerances and of changing from double to quadruple precision floating point arithmetic. It is observed that refining the integration tolerance might extend the interval of topological stability. However, the dynamics of the system remain chaotic and the solutions will eventually diverge for sufficiently long times.

\begin{figure}
	\centering
	\includegraphics[width=\linewidth]{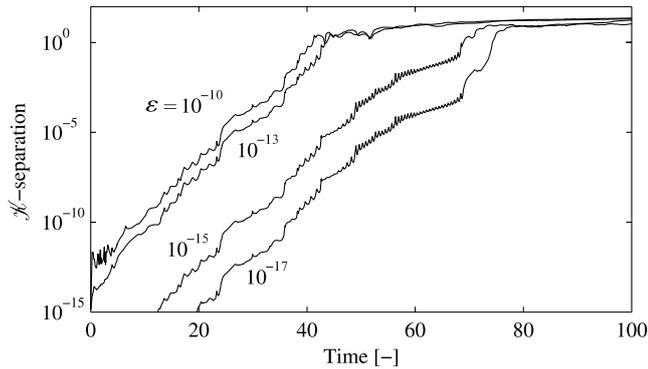}
	\caption{$\mathscr{K}$-separation for the four-body problem for different integration tolerances. The solutions for $\varepsilon=10^{-15}$ and $10^{-17}$ are computed in quadruple precision floating point arithmetic.\label{Fig:Ksep_4stars_tols}}
\end{figure}

\section{Conclusions}

The topology of the KS transformation has important consequences in the stability and accuracy of the solutions in KS space. There are two key aspects to consider when studying the stability of the motion. First, the presence of a fictitious time that replaces the physical time as the independent variable. Second, the dimension-raising nature of the Hopf fibration.

Classical theories of stability are based on the separation between nearby trajectories. Having introduced a fictitious  time, the question on how to synchronize the trajectories arises. The numerical stabilization of the equations of motion by KS regularization relates to solutions synchronized in fictitious time. But a theory of stability synchronized in physical time allows the translation of concepts such as attractive sets, Lyapunov stability, Poincar\'e maps, etc. to KS language.

The additional dimension provides a degree of freedom to the solution in parametric space. In general the free parameter can be fixed arbitrarily, with little or no impact on the resulting trajectory in Cartesian space. However, as strong perturbations destabilize the system, different values of the free parameter may result in completely different solutions in time. This phenomenon is caused by numerical errors and destroys the topological structure of KS transformation: points in a fiber are no longer transformed into one single point. By monitoring the topological stability of the integration it is possible to estimate an indicator similar to the Lyapunov time.


\section*{Acknowledgements}

This work is part of the research project entitled ``Dynamical Analysis, Advanced Orbit Propagation and Simulation of Complex Space Systems'' (ESP2013-41634-P) supported by the Spanish Ministry of Economy and Competitiveness. Authors thank the Spanish Government for its support. J. Roa specially thanks ``La Caixa'' for his doctoral fellowship and S. Le Maistre for motivating him to work on $N$-body problems.








\appendix

\section{Orthogonal bases}\label{Sec:Appendix}

\subsection{Basis attached to the fiber}\label{Sec:basis_fiber}
In Section~\ref{Sec:velocity} it is shown that fibers are circles in $\mathbb{U}^4$. Let $\myvec{u}$ and $\myvec{w}=\mathscr{R}(\vartheta;\myvec{u})$ be two vectors attached to a fiber $\mathcal{F}$. They span a plane containing the fiber. This plane is not a plane of the Levi-Civita type\footnote{A plane of the Levi-Civita type, or $\mathscr{L}$-plane, is a plane spanned by two vectors satisfying the bilinear relation $\ell(\myvec{u},\myvec{w})=0$. Planes of this type are KS-transformed to planes in $\mathbb{E}^3$, and the mapping is conformal: angles are doubled and distances to the origin are squared \citep[][pp.~273--276]{stiefel1971linear}.} because $\ell(\myvec{u},\myvec{w})\neq0$. Since trajectories intersect fibers at right angles this subspace is transversal to the flow. An orthogonal basis can be attached to the resulting plane, allowing projections on the transversal subspace. Although arbitrary orthonormal bases can be constructed via the Gram-Schmidt procedure \citep[][pp.~529--530]{nayfeh2004applied}, the basis described in this section appears naturally in the formulation.

Associated to every vector $\myvec{u}$ there is a KS matrix $\mymatrix{L}(\myvec{u})$. The columns of the matrix define a vector basis $\mathfrak{B}=\{\myvec{u}^1,\myvec{u}^2,\myvec{u}^3,\myvec{u}^4\}$, with $\myvec{u}^1\equiv\myvec{u}$. The basis $\mathfrak{B}$ is orthogonal,
\begin{equation}
	\myvec{u}^i\cdot \myvec{u}^j = r\delta_{ij}
\end{equation}
Here $\delta_{ij}$ denotes Kronecker's delta. Assuming $x\geq0$ every point in the fiber generated by $\myvec{w}=\mathscr{R}(\vartheta;\myvec{u})$ lies in the plane spanned by $\myvec{u}^1$ and $\myvec{u}^4$, i.e. $\myvec{w}\cdot\myvec{u}^2=\myvec{w}\cdot\myvec{u}^3=0$ for all $\vartheta$. Conversely, for $x<0$ the fiber is confined to the $\myvec{u}^2\myvec{u}^3$ plane. The basis $\mathfrak{B}$ is an orthogonal basis attached to the fiber at $\myvec{u}$. In addition,
\begin{equation}
	\myvec{u}^4\cdot\myvec{u}^\prime = \ell(\myvec{u},\myvec{u}^\prime) = 0
\end{equation}
meaning that $\myvec{u}^\prime$ is perpendicular to $\myvec{u}^4$. In fact, $\myvec{u}^4=-\myvec{t}$, as shown by Eq.~\eqref{Eq:tangent}.

The vectors arising from the products $\mymatrix{L}(\myvec{u})\myvec{u}^i$, $i=1,2,3$ have a vanishing fourth component. They are equivalent to vectors in $\mathbb{E}^3$. However, the fourth component of the product $\mymatrix{L}(\myvec{u})\myvec{u}^4$ is not zero. The three vectors obtained by these transformations correspond to the position vector $\myvec{r}$, and a pair of orthogonal vectors spanning the plane tangential to the 2-sphere at $\myvec{r}$ in $\mathbb{E}^3$. These vectors are the columns of the associated Cailey matrix.

\subsection{Cross product}
\citet[][pp.~277--281]{stiefel1971linear} sought a definition of cross product in the parametric space $\mathbb{U}^4$ when discussing the orthogonality conditions of vectors and Levi-Civita planes. Although the cross product of two vectors in $\mathbb{R}^3$ is intuitive, its generalization to higher dimensions is not straightforward. Independent proofs from different authors \citep[see for example][]{brown1967vector} show that the cross product of two vectors only exists in dimensions $1,3,7$; for $n$ dimensions the cross product involves $n-1$ vectors. \citet{stiefel1971linear} defined the product $\myvec{p}=\myvec{u}\times\myvec{v}$ as
\begin{equation}
	\myvec{p} = \mymatrix{L}(\myvec{u})\,\myvec{v}^4
\end{equation}
where $\myvec{v}^4=(v_4,-v_3,v_2,-v_1)^\top$ is the fourth column of $\mymatrix{L}(\myvec{v})$. The properties of this construction motivated the authors to call $(p_1,p_2,p_3)$ the cross product of $\myvec{u}\times\myvec{v}$, with $p_4=\myvec{u}\cdot\myvec{v}$. In the following lines we analyze in more detail this construction and connect with alternative definitions provided by \citet{vivarelli1987geometrical} and \citet{deprit1994linearization}.


Let $\{\myvec{e}_1,\myvec{e}_2,\ldots,\myvec{e}_n\}$ be an orthogonal basis in $\mathbb{R}^n$. The Grassmann exterior product gives rise to the bivectors $\myvec{e}_i\wedge\myvec{e}_j$, trivectors $\myvec{e}_i\wedge\myvec{e}_j\wedge\myvec{e}_k$, and successive blades of grade $m\leq n$ \citep[][\S II]{flanders1989differential}. They constitute the subspaces $\bigwedge^m\mathbb{R}^n$ of the exterior algebra:
\begin{equation}
	\bigwedge\mathbb{R}^n= \mathbb{R} \,\oplus\,\mathbb{R}^n \,\oplus\, \bigwedge^2\mathbb{R}^n \,\oplus\,\ldots\,\oplus\,\bigwedge^n\mathbb{R}^n   
\end{equation}
noting that $\bigwedge^0\mathbb{R}^n=\mathbb{R}$ and $\bigwedge^1\mathbb{R}^n=\mathbb{R}^n$. Without being exhaustive we simply recall that such exterior algebra is associative with unity, satisfying $\myvec{e}_i\wedge\myvec{e}_j=-\myvec{e}_j\wedge\myvec{e}_i$ and $\myvec{e}_i\wedge\myvec{e}_i=0$. The exterior product of two parallel vectors vanishes. We shall write $\myvec{e}_{ij\ldots k}=\myvec{e}_i\wedge\myvec{e}_j\wedge\ldots\wedge\myvec{e}_k$ for brevity.

In Section~\ref{Sec:basis_fiber} an orthogonal basis attached to $\myvec{u}$ was defined, where two of its vectors are KS-transformed to vectors spanning the plane tangent to the 2-sphere in $\mathbb{E}^3$. Identifying $\myvec{u}^i=\sqrt{r}\,\myvec{e}_i$, the exterior product of vectors $\{\myvec{u}^1,\myvec{u}^2,\myvec{u}^3,\myvec{u}^4\}$ generates the oriented hypervolume
\begin{equation}
	\myvec{u}^1\wedge\myvec{u}^2\wedge\myvec{u}^3\wedge\myvec{u}^4 = -r^2 \myvec{e}_{1234}  
\end{equation}
provided that $\det(\mymatrix{L}(\myvec{u}))=-r^2$. In three dimensions the exterior product is equivalent to the cross product, given $\myvec{e}_1\times\myvec{e}_2=\myvec{e}_3$, $\myvec{e}_1\times\myvec{e}_3=-\myvec{e}_2$ and $\myvec{e}_2\times\myvec{e}_3=\myvec{e}_1$. Applying the cross product to the first three elements of $\mathfrak{B}$ provides
\begin{equation}
	\myvec{u}^1\times\myvec{u}^2\times\myvec{u}^3 = r \myvec{u}^4
\end{equation}
This result confirms that $\mathfrak{B}$ is, indeed, an orthogonal basis. 

\citet{vivarelli1987geometrical} and \citet{deprit1994linearization} worked in the more general Clifford algebra $\cliff_3$. Introducing the Clifford product of two vectors
\begin{equation}
	\myvec{a}\myvec{b} = \myvec{a}\cdot\myvec{b} + \myvec{a}\wedge\myvec{b}
\end{equation}
the exterior algebra over $\mathbb{R}^3$ can be identified with the Clifford algebra $\cliff_3$: bivectors and trivectors become $\myvec{e}_i\wedge\myvec{e}_j\to\myvec{e}_i\myvec{e}_j$ and $\myvec{e}_1\wedge\myvec{e}_2\wedge\myvec{e}_3\to\myvec{e}_1\myvec{e}_2\myvec{e}_3$. Note that $\myvec{b}\myvec{a} = \myvec{a}\cdot\myvec{b} - \myvec{a}\wedge\myvec{b}\neq\myvec{a}\myvec{b}$, so the Clifford product is not commutative. Thus, the even subalgebra $\cliff_3^+=\mathbb{R}\oplus\bigwedge^2\mathbb{R}^3$, isomorphic to the quaternion algebra $\mathbb{H}$, is not commutative either. 

The algebra $\mathbb{H}\simeq\cliff_3^+$ is endowed with the multiplication rules $\myvec{e}_1^2=\myvec{e}_2^2=\myvec{e}_3^2=-1$, and the product of vectors is anticommutative, $\myvec{e}_{i}\myvec{e}_j=-\myvec{e}_j\myvec{e}_i$. Identifying these bivectors with the quaternion basis elements
\begin{equation}
	\myvec{e}_2\myvec{e}_3 = \mathrm{i},\quad \myvec{e}_3\myvec{e}_1 = \mathrm{j},\quad \myvec{e}_1\myvec{e}_2 = \mathrm{k}
\end{equation}
the product of two quaternions $\myquat{u}$ and $\myquat{v}$ is established. \citet{vivarelli1987geometrical} rewrote \citeauthor{stiefel1971linear}'s form of the cross product in terms of the quaternion product
\begin{equation}
	\myquat{u}\times\myquat{v} = \frac{1}{2}(\myquat{u}\mathrm{k}\myquat{v}_\ast - \myquat{v}\mathrm{k}\myquat{u}_\ast)
\end{equation}
where $_\ast$ denotes the involution $\myquat{u}_\ast=u_1+u_2\mathrm{i}+u_3\mathrm{j}-u_4\mathrm{k}$. Disregarding the arrangement of the components, \citet{deprit1994linearization} defined the cross product of two quaternions as
\begin{equation}
	\myquat{u}\times\myquat{v} = \frac{1}{2}(\myquat{v}\myquat{u}^\dagger-\myquat{u}\myquat{v}^\dagger)
\end{equation}
Here $^\dagger$ denotes the quaternion conjugate. The difference between these quaternionic definitions and the original one from \citeauthor{stiefel1971linear} is the fact that $\myquat{u}\times\myquat{v}$ is a pure quaternion, i.e. $\Re(\myquat{u}\times\myquat{v})=0$, whereas the fourth component of \citeauthor{stiefel1971linear}'s product $\myvec{u}\times\myvec{v}$ is $\myvec{u}\cdot\myvec{v}$.




\bsp	
\label{lastpage}
\end{document}